\documentclass[a4paper,UKenglish,cleveref,autoref,thm-restate]{lipics-v2021-mod}

\hideLIPIcs
\nolinenumbers

\title{Warm-Starting All-Pairs Shortest Paths with Predictions}

\author{Adam Polak}{Bocconi University, Milan, Italy \and \url{https://adampolak.github.io/}}{adam.polak@unibocconi.it}{https://orcid.org/0000-0003-4925-774X}{}

\author{Jonas Schmidt}{Bocconi University, Milan, Italy \and \url{https://schmidtjonas.github.io/}}{jonas.schmidt2@phd.unibocconi.it}{https://orcid.org/0000-0002-1115-3868}{}

\authorrunning{A. Polak and J. Schmidt}

\Copyright{Adam Polak and Jonas Schmidt}

\supplement{\url{https://github.com/adampolak/warm-starting-apsp}}

\funding{Funded by the Italian Ministry of University and Research (MUR) under Ministerial Decree No.~18169 of 17 November 2025 -- FIS 3 Call CUP: J53C25002160001}

\acknowledgements{We thank Antonios Antoniadis and Yasamin Nazari for helpful discussion.}

\usepackage{amsfonts}
\usepackage{pgfmath}
\usepackage{mathtools}
\usepackage{comment}

\usepackage[ruled,vlined,linesnumbered]{algorithm2e}
\SetKwInput{KwInput}{Input}
\SetKwInput{KwOutput}{Output}
\DontPrintSemicolon
\SetKwComment{Comment}{// }{}
\SetArgSty{upshape}
\SetKwProg{Def}{fn}{:}{}

\SetKwFunction{FComputeEstimates}{ComputeEstimates}
\SetKwFunction{FVerifyEstimates}{VerifyEstimates}
\SetKwFunction{FFixMistakes}{FixMistakes}
\SetKwFunction{FFullAlg}{FullAlgorithm}

\DeclareMathOperator*{\otilde}{\Tilde{\mathcal{O}}} 
\DeclareMathOperator*{\bigo}{\mathcal{O}} 

\usepackage{amssymb}
\newcommand{\domprod}{\mathbin{\text{\textcircled{$<$}}}}

\newtheorem{problem}[theorem]{Problem}
\crefname{problem}{problem}{problems}

\crefname{enumi}{}{} 

\newtheorem{open}{Open Problem}

\DeclareMathOperator{\jumpOp}{U}
\DeclareMathOperator{\detourOp}{detour}

\newcommand*{\detour}[4]{\detourOp_{{#1}, {#2}, {#3}}({#4})}
\newcommand*{\jump}[3]{\jumpOp({#1}, {#2}, {#3})}
\newcommand*{\err}[3]{\eta({#1}, {#2}, {#3})}

\begin{document}

\maketitle

\begin{abstract}
One of the three key hypotheses of fine-grained complexity asserts that computing All-Pairs Shortest Paths (APSP) requires cubic time, up to subpolynomial factors, in the worst case. We initiate the study of APSP in the paradigm of algorithms with predictions, also known as learning-augmented algorithms. We propose an APSP algorithm that takes as additional input a \emph{prediction} (e.g., given by a model learned from similar instances seen in the past) consisting of sets of vertices causing the shortest \emph{detour} for each pair of vertices. The algorithm runs in time $\mathcal{O}(n^{2.83} + \eta n)$, where $\eta$ denotes the \emph{prediction error} defined as the number of pairs of vertices for which, informally speaking, the prediction was not sufficient to compute and certify optimality of the shortest path length. This is already subcubic when the prediction error is (polynomially) smaller than its maximum possible values $n^2$, i.e., whenever the prediction is at least slightly better than terrible.

We build on the co-nondeterministic algorithm for the Exact Triangle problem by Chan, Vassilevska Williams, and Xu (STOC 2023), essentially enabling this algorithm to detect mistakes in the nondeterministic certificate and recover from them.

Our result constitutes the first necessary step towards designing learning-augmented algorithms for problems with known fine-grained lower bounds conditioned on the APSP Hypothesis.
\end{abstract}

\section{Introduction}

Given a directed edge-weighted graph $G=(V,E,w)$, the All-Pairs Shortest Paths (APSP) problem asks to compute for each pair of vertices $u, v \in V$ the distance from $u$ to $v$, i.e., the minimum total weight of edges on a path from $u$ to $v$. In $n$-vertex graphs, APSP can be easily solved in time $\bigo(n^3)$, e.g., using the Floyd--Warshall algorithm~\cite{Roy59,Floyd62a,Warshall62}, or, when the edge weights are non-negative, by running Dijkstra's algorithm~\cite{Dijkstra59} from each vertex. The fastest known APSP algorithm~\cite{Williams18} shaves off only a factor of $2^{\bigo(\sqrt{\log n})}$, and the APSP Hypothesis~\cite{VWW18} asserts that the cubic running time is optimal up to subpolynomial factors. It is one of the three main hypotheses of fine-grained complexity~\cite{VW18}, alongside the 3SUM Hypothesis and Strong Exponential Time Hypothesis (SETH).

Imagine we solve many instances of APSP that are in some sense \emph{similar} to each other, e.g., many snapshots of a road network taken over time with edge weights representing estimated travel times that get updated between the snapshots.

\begin{center}\itshape
Can we use a common structure shared by similar APSP instances \\ in order to solve them faster?
\end{center}

This question fits into the recent line of research on learning-augmented algorithms, also known as algorithms with predictions, see, e.g., \cite{MitzenmacherV20,LindermayrM22}.
In particular, it is the kind of question asked in papers on \emph{warm-starting} algorithms (e.g., \cite{DinitzILMV21,SakaueO22a,BaiC23}). There, the algorithm's input is enriched with a \emph{hint} or \emph{prediction}, produced, e.g., by a model trained on previously solved similar instances. The algorithm must remain correct unconditionally, but its running time is allowed to depend smoothly on the quality of the prediction -- quantified by an error measure denoted by $\eta$. In particular, even a mildly informative prediction should already yield a provable speedup.

Many warm-starting results to date address problems that already admit algorithms running in almost-linear time, e.g., Maximum Flow~\cite{DaviesMVW23,PolakZ24,DaviesVW24}, Negative-Weight Single-Source Shortest Paths~\cite{ChenSVZ22,LattanziSV23}, and Sorting~\cite{BaiC23}.
There, potential improvements from predictions are limited to shaving subpolynomial factors. APSP is qualitatively different: the best worst-case running times remain essentially cubic, and any provable polynomial improvement over $n^3$ is of independent interest. This makes APSP a natural test case for understanding whether predictions can help circumvent fine-grained barriers without giving up on worst-case correctness.
Moreover, there are many problems with fine-grained lower bounds conditional on the APSP Hypothesis (see, e.g., \cite{VWW18,VW18}). Consequently, a learning-augmented APSP algorithm can be viewed as a prerequisite for obtaining learning-augmented algorithms for those problems.

\subparagraph{What does it mean for APSP instances to be similar?}
If the input graph changes only slightly between instances, the natural framework to work with is dynamic algorithms, which update previously computed distances after each change (insertion or deletion of a vertex or an edge). The literature on dynamic algorithms for APSP is vast. In particular, Mao~\cite{Mao24a} gives a fully dynamic APSP algorithm with $\otilde(n^{2.5})$ worst-case vertex update time, which is believed to be likely optimal~\cite{AbrahamCK17}. This translates to subcubic total time as long as the number of input changes is much less than $\sqrt{n}$, so this regime is not the focus of the present work.

Instead, we consider a different notion of similarity in which edge weights may change significantly, yet the structure of (many) shortest paths remains stable. A prototypical example is a road network: vertices represent junctions, edges represent road segments, and the weight of an edge encodes an estimated travel time that varies with traffic. Between two timestamps the weights can fluctuate widely, but for many origin--destination pairs the fastest route still passes through essentially the same intermediate junctions. Our goal is to exploit this kind of stability.

\subparagraph{Co-nondeterministic algorithms, or what shall be predicted?}
A popular choice of prediction for warm-starting algorithms is simply a predicted solution (e.g.,~\cite{DinitzILMV21,ChenSVZ22,DaviesMVW23,PolakZ24,DaviesVW24}). A difficulty we face is that, for APSP, verification is known to be essentially as hard as computation. More precisely, already the problem of deciding whether a presumed distance matrix satisfies the triangle inequality admits a truly subcubic algorithm if and only if APSP does~\cite{VWW18}. In other words, even if our learning-augmented APSP algorithm was provided with a presumably untrusted distance matrix that happened to be perfectly correct, it could not know it without first spending cubic time on verification. A useful prediction must therefore include auxiliary information that enables fast verification, and ideally also enables the algorithm to recover when only a small part of the prediction is wrong.\footnote{This also distinguishes our work from learning-augmented \emph{approximation} algorithms, which too use predicted solutions, and manage to beat conditional lower bounds based on $\mathrm{P}\ne \mathrm{NP}$ or the Unique Games Conjecture, but do not need to verify correctness of predictions~\cite{ergun2022learningaugmented,aamand25improved,AEPV25,ghoshal25constraint}.}

This naturally leads us to \emph{co-nondeterministic} algorithms. Such algorithms, running in truly subcubic time, are known for APSP~\cite{CarmosinoGIMPS16,ChanWX23}. Informally, they assert the existence of short witnesses whose validity can be quickly checked. In our setting, the prediction plays the role of such a witness. However, not every co-nondeterministic approach is suitable for us. For instance, the construction of~\cite{CarmosinoGIMPS16}, which relies on hashing after a reduction to the so-called Exact Triangle problem, does not appear to degrade gracefully under small inaccuracies in the witness. Instead, we build on the alternative framework based on Fredman's trick and dominance product~\cite{ChanWX23}, which was originally proposed to handle real weights.

\subsection{Our contribution}

We propose an APSP algorithm that takes as prediction a \emph{certificate} consisting of sets of vertices causing the shortest \emph{detour} for each pair of vertices. The algorithm runs in time $\bigo(n^{2.83} + \eta n)$, where $\eta$ denotes the \emph{prediction error} defined as the number of pairs of vertices for which, informally speaking, the prediction was not sufficient to compute and certify optimality of the shortest path length. In this section we define all these notions formally.

\subparagraph*{Certificate.}
For an $n$-vertex weighted graph $G=(V,E,w_G)$, a \textbf{\emph{certificate}} is a matrix $C \in \mathcal{P}(V)^{V\times V}$ that contains, for every $u,v \in V$, a set $C[u,v] \subseteq V$ of exactly $q$ vertices, where $q \le n$ is a parameter to be chosen. The certificate hence has size $\sum_{u,v \in V(G)}|C[u,v]| = n^2q$.

This describes the form of a certificate. Let us now give some intuition behind its semantics. In a ground-truth certificate, each entry $C[u,v]$ should be the set of ``best intermediate vertices'' for a path from $u$ to $v$ in $G$.
More formally, for a $V \times V$ matrix $D$, we define the \textbf{\emph{detour}} of a vertex $w$ as
\[\detour{D}{u}{v}{w} \coloneqq D[u,w] + D[w,v] - D[u,v].\]
Since later we work with estimated distances, which correspond to lengths of paths that are not necessarily shortest, we define the detour for general matrices. Ideally, the certificate entry $C[u, v]$ should contain vertices $w$ with the $q$ smallest values of $\detour{D^*}{u}{v}{w}$, where $D^*$ is the true shortest paths matrix of $G$.

We note that breaking ties is not important here, since we later only need to make sure that vertices with strictly smaller detour than the ones in the certificate are also in the certificate. Hence, one can break ties, for example, by index.
Clearly, one can compute a ground-truth certificate in time $\bigo(n^3)$ with any APSP algorithm.

\subparagraph*{Distance estimate matrix.}
Morally, our error measure tries to capture the impact of inaccuracies in the certificate $C$ that lead the verification algorithm to fail verifying the true shortest paths matrix $D^*$.
However, since the algorithm does not know $D^*$, we define the error based on an estimate of $D^*$ that is guided by the certificate.
For a graph $G$ and certificate $C$, we define the corresponding \textbf{\emph{distance estimate matrix}} $D_{G,C}$ to be the element-wise maximum matrix that satisfies, for all $u,v \in V$, both
\begin{enumerate}
    \item $d(u,v) \le D_{G,C}[u,v]$, and
    \item $D_{G,C}[u,v] = \min \{w_G(u,v), \min_{w \in C[u,v]} D_{G,C}[u,w] + D_{G,C}[w,v]\}$.
\end{enumerate}
When $G$ and $C$ are clear from context, we drop the subscripts and write just $D$.

To put this definition in context, consider the classical Floyd--Warshall algorithm for APSP~\cite{Roy59,Floyd62a,Warshall62}. At a high level, it works by repeatedly updating the shortest path between some vertex pair $(u,v)$, i.e., a matrix entry $M[u,v]$, by considering a detour via another vertex $w$, i.e., $M[u,w] + M[w,v]$. We call such an update a \emph{relaxation of $M[u,v]$ with $w$}.
Remember that a good certificate $C[u,v]$ is supposed to contain vertices $w$ such that $D^*[u,w] + D^*[w,v]$ is close to $D^*[u,v]$. If we trust the certificate, it hence does not make sense to relax $M[u,v]$ with vertices $w$ that are not present in $C[u,v]$.

The matrix $D_{G,C}$ is the best (i.e., element-wise minimum) matrix achievable when initializing $M$ with edge weights and repeatedly performing only relaxations as indicated by the certificate, that is, relaxing $M[u,v]$ only with vertices $w \in C[u,v]$.
Therefore, if $C[u,v]$ contains at least one intermediate vertex on a shortest path from $u$ to $v$ in $G$ for every $u, v \in V$ (with a shortest path of at least two edges), then $D_{G,C}$ is the correct shortest paths matrix for $G$, i.e., $D_{G,C}[u,v] = d(u, v)$, for every $u, v \in V$.
In Section~\ref{sec:distance-estimates} we show how to compute $D_{G,C}$ efficiently, in time $\bigo(n^2(q+\log n))$, i.e., proportional to the number of allowed relaxations $n^2q$.
Surprisingly, this requires moving away from the intuition from Floyd--Warshall and towards a Dijkstra-like order of relaxations.

\subparagraph*{Error definition.}
Now we define how we measure how ``far'' the predicted certificate (given in the input to our algorithm) is from a ground-truth certificate.
Morally, the prediction error $\eta$ counts the number of pairs $(u,v) \in V^2$, for which there is a vertex $w$ with small $\detour{D^*}{u}{v}{w}$ that is missing from the certificate $C[u,v]$, where $D^*$ denotes the true shortest paths matrix.
These are exactly the entries $D^*[u,v]$ for which $C[u,v]$ does not provide enough information to verify their optimality.
Since we can only verify $D_{G,C}$ instead of $D^*[u,v]$, we define the error based on the estimate matrix.

For a graph $G$, a certificate $C$, and a parameter $p \le q$, we define the set $\jump{G}{C}{p}$ of \textbf{\emph{unverifiable pairs}} $(u,v)$ for which there is a vertex $w$ not in the certificate $C[u,v]$ that is among the $p$ vertices with the smallest $\detour{D}{u}{v}{w}$ values.
Formally,
\[ \jump{G}{C}{p} \coloneqq \big\{(u,v) \colon \exists w \in V \setminus C[u,v]\colon |\{c \in C[u,v] \colon \detour{D}{u}{v}{c} \le \detour{D}{u}{v}{w}\}| < p\big\}. \]

Now we are ready to define the \textbf{\emph{prediction error}} measure
\[\err{G}{C}{p} \coloneqq \big|\,\{(u,v) \colon D_{G,C}[u,v] \ne d_{G}(u,v)\}\,\cup\,\jump{G}{C}{p}\,\big|,\]
or just $\eta$ when the parameters are clear from the context.
Our main theorem reads as follows.

\begin{restatable}{theorem}{mainThm} \label{thm:full-alg}
    There is an algorithm that, given an $n$-vertex graph $G = (V,E,w_G)$, a certificate $C$ with $|C[u,v]| = q \le n$ for all $u,v \in V$, and a parameter $p \le q$, computes APSP in $G$ in time \[\bigo(n^2(q + \log n) + n^{3.66} / p + \err{G}{C}{p} \cdot n).\]
    For $p = \Theta(n^{0.83})$ and $q = \Theta(p)$ it runs in time $\bigo(n^{2.83} + \eta n)$.
\end{restatable}

For a fixed $q$, the parameter $p$ can be used as a trade-off between error and verification time. 
For $p' \le p$, notice that $\jump{G}{C}{p'} \subseteq \jump{G}{C}{p}$, so also $\err{G}{C}{p'} \le \err{G}{C}{p}$.
In other words, a smaller value of $p$ leads to a smaller prediction error but a longer verification time.
In the case when the error value $\eta$ is small enough not to dominate the running time, it is reasonable to choose $p$ and $q$ to balance the remaining running time. This leads to the choice of $p = \Theta(n^{0.83})$, and gives the claimed running time guarantee.

\begin{algorithm}[b]
    \caption{The full algorithm, combining all three stages}\label{alg:combined}
    \KwInput{Graph $G$, certificate $C$, and parameter $p$}
    \KwOutput{Distance matrix $D^*$}
    \Def{\FFullAlg{$G$, $C$, $p$}}{
        $D \gets$ \FComputeEstimates{$G$, $C$} \tcp*{Stage 1 (see \Cref{alg:distance-estimates})}
        $U \gets$ \FVerifyEstimates{$D$, $C$, $p$} \tcp*{Stage 2 (see \Cref{alg:verify-cert})}
        $D^* \gets$ \FFixMistakes{$D$, $U$} \tcp*{Stage 3 (see \Cref{alg:fixing})}
        \Return{$D^*$}
    }
\end{algorithm}

We remark that the intricacies of our error definition seem hard to avoid with our current understanding of the APSP problem.
Indeed, as discussed earlier, any learning-augmented algorithm has to degenerate to a co-nondeterministic algorithm when the prediction is perfect.
The error measure hence has to capture the amount of additional work imposed on the verification algorithm by inaccuracies in the prediction.
Given that there are only two known co-nondeterministic algorithms for APSP~\cite{CarmosinoGIMPS16,ChanWX23}, out of which one is inherently ``non-smooth'' due to integer hashing~\cite{CarmosinoGIMPS16}, we would be surprised if the intricacies in the prediction error definition were completely avoidable.
In \Cref{sec:open}, we discuss minor simplifications of the error measure that could be within reach.

On the positive side, we highlight that our prediction error only depends on \emph{how many} entries of the distance estimate matrix are wrong or unverifiable, and not by \emph{how much} they are off. This is in contrast to many learning-augmented results that concern the $\ell_1$ prediction error~\cite{DaviesMVW23,PolakZ24,MitzenmacherV20,DinitzILMV21}.
Our error measure is closer in spirit to the $\ell_0$ error, which is generally smaller but harder to use~(see, e.g., \cite{ChenSVZ22}).
Employing this kind of error in our case seems possible due to two factors: (1) since verifying correctness of a presumed solution is part of our algorithm, we exactly spot where prediction errors lie, and (2) thanks to the Dijkstra-like order in which we fix these errors, we are able to guarantee correctness of an entry once we recalculate it for the first time -- no matter how far off it was.
Without this specific order, we would only be guaranteed to decrease an entry by at least 1 in a single recomputation, yielding a running time dependent on the $\ell_1$ error $\sum_{u,v}D[u,v] - D^*[u,v]$.

In Section~\ref{sec:empirical} we give some empirical evidence that one can hope for $\eta \ll n^2$ in practice.

In the remaining subsection, we go into more detail on how to compute and verify the distance estimate matrix, as well as how to order and fix wrong entries.

\subparagraph*{Algorithm outline.}

Our algorithm proceeds in three stages (see \Cref{alg:combined}). In the first stage, we compute the distance estimate matrix $D$ from a given certificate.
Remember that this corresponds to the entry-wise minimum matrix achievable by starting from $D[u,v] = w_G(u,v)$ and repeatedly relaxing entries $D[u,v]$ with vertices $w \in C[u,v]$.
A natural idea how to compute such a matrix would be running the Floyd--Warshall algorithm but performing only relaxations that are indicated by the certificate.
However, this approach fails, as shown by the simple counterexample in \Cref{fig:floyd-example}.

\begin{figure}
    \centering
    \includegraphics[width=0.35\linewidth,page=2]{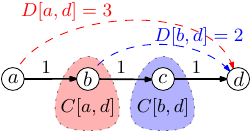}
    \caption{A counterexample for why the Floyd--Warshall algorithm with relaxations restricted to the certificate cannot compute $D_{G,C}$. In the given graph $G$, with four vertices and three black edges, each of weight 1, with $C$ defined by $C[a,d] = \{b\}$ and $C[b,d]=\{c\}$ (all other certificate sets equal $\{a\}$), $D_{G,C}$ should contain distances according to the colored arrows. However, when relaxing first $D[a,d]$ with $b$ and then $D[b,d]$ with $c$, only the latter blue distance is correct. Note that standard Floyd--Warshall corresponds to the case where $C[u,v] = V$, for every $u,v \in V$, where the order of considered intermediate vertices does not matter.}
    \label{fig:floyd-example}
\end{figure}

We circumvent this issue by changing the order of relaxations in Floyd--Warshall.
Doing so requires two additional ideas:
First, while we have to be careful with changing the order of relaxations in Floyd--Warshall~\cite{koo24},
a sufficient condition for correctness of a reordering of relaxations is guaranteeing that any accessed entry already contains the correct distance.
Second, we are in fact able to guarantee this condition by turning all edge weights non-negative with the well known Johnson's trick~\cite{Johnson77} from negative shortest path algorithms and relaxing in increasing order of distances.
Surprisingly, the resulting algorithm at first glance more closely resembles Dijkstra's algorithm than Floyd--Warshall.
To simplify the presentation, we perform Johnson's trick once in the beginning as discussed in Section~\ref{sec:negative-weights} and restrict our view to non-negative weights.

The second stage is responsible for verification of the distance estimates, that is, detecting entries $D[u,v]$ that could be relaxed further if we do not restrict to relaxations indicated by $C$. 
While the matrix $D$ already satisfies \begin{equation}D[u,v] \le \min_{\mathclap{w \in C[u,v]}} D[u,w] + D[w,v],\end{equation} we would need \begin{equation} D[u,v] \le \min_{w \in V} D[u,w] + D[w,v] \end{equation} to make sure that no further relaxation is possible.
The verification algorithm in the second stage finds all entries that could be relaxed as well as entries for which the certificate is not sufficient to verify (2).
To achieve that, we adapt the co-nondeterministic algorithm for Exact Triangle from Theorem 9.7 in~\cite{ChanWX23}.
Intuitively, our certificate takes the role of their nondeterministic prover.
The algorithm for this stage computes a set $U \subseteq \jump{G}{C}{p}$, such that all $(u,v) \in V^2 \setminus U$ are guaranteed to satisfy (2).

In the third stage, we use the set $U$ as a starting point to fix all incorrect distance entries in $D$.
We face a similar problem as in Stage 1 of our algorithm: This time, we want to restrict ourselves to $O(n)$ relaxations per incorrect or unverified entry in $D$ while guaranteeing that afterwards each entry satisfies (2).
To solve it, we again update distance entries Dijkstra-like in the increasing order, relying on the fact that we can assume non-negative edge weights.
Additionally, the algorithm works in a lazy manner and only relaxes entries $D[u,v]$ with vertices $w$ if at least one of $(u,v)$, $(u,w)$, or $(w,v)$ is in $U$ or has been witnessed as incorrect.
Our algorithm for this section can hence be regarded as a more sophisticated implementation of the ideas already employed in the first stage.

\subsection{Learnability}

How do we argue that it is possible to efficiently learn a certificate with a small prediction error? A standard approach (see, e.g., \cite{DinitzILMV21,ChenSVZ22,PolakZ24}) would be to use the \emph{probably approximately correct} (PAC) learning framework. There, we assume we have sample access to an unknown distribution $\mathcal{D}$ over $V \times V$ matrices of edge weights. The goal is to find a certificate that (approximately) minimizes the expected prediction error $\mathbb{E}_{G \sim D} \; \err{G}{C}{p}$. By a standard argument (see, e.g.,~\cite[Section~4]{PolakZ24}), since the number of different possible certificates is upper bounded by $\binom{n}{q}^{n^2} \le 2^{n^3}$, it suffices to consider $k=O(n^3)$ samples from $\mathcal{D}$, and solve the corresponding empirical risk minimization (ERM) problem, i.e., find a certificate that minimizes the prediction error averaged over the sample:

\begin{problem}[APSP Certificate Learning]\label{prob:learning}
    Given $k$ APSP instances $G_1, \ldots, G_k$ on the same vertex set $V$, and parameters $p$, $q$, output a certificate $C$, with $|C[u,v]| = q$ for every $u,v \in V$, that minimizes the average prediction error $\frac{1}{k}\sum_{i=1}^k \err{G_i}{C}{p}$.
\end{problem}

In Section~\ref{sec:hardness}, we show that this problem is, unfortunately, NP-hard to solve exactly. However, in our context, a constant-factor approximation algorithm would be entirely sufficient, because a constant-factor loss in the prediction error translates to a constant-factor loss in the running time of our APSP algorithm, which anyway vanishes in the big-O notation. We leave it as open problem to find such an approximation algorithm (see \Cref{sec:open}).

\subsection{Open problems}\label{sec:open}

We state two important open problems we are left with. The first is about improving and generalizing our error measure, the second is about learning certificates.

\subparagraph{Cleaner prediction error measure.}

A pair of vertices $(u,v)$ can contribute to our error measure $\eta$ in two ways.
The first way is when the computed distance estimate $D[u,v]$ from Stage 1 is not the correct distance.
The second way is when we cannot verify the distance estimate $D[u,v]$ in the second stage of our algorithm.

The second term seems unavoidable, because we need a way to beat cubic time while verifying correctness of the distance matrix. To our best knowledge, the co-nondeterministic algorithm from~\cite{ChanWX23} is the only suitable one among known ways to achieve that.
However, we see no apparent reason why it should not be possible to remove the first term entirely from the error.

\begin{open}
    Can we remove the term $|\{(u,v) \colon D[u,v] \ne d(u,v)\}|$ from the definition of the prediction error $\eta$?
\end{open}

Notice that if the first term is significantly larger than the second term, there is an incorrect entry $D[u,v]$ such that many shortest paths contain a subpath from $u$ to $v$.
One approach to solve such a case could be to rerun the first stage of our algorithm once an incorrect entry is corrected. The correction could then propagate to all superpaths, without spending linear time for each of them.
Making such an approach work formally could require phrasing the error in a more direct way without a distance estimate matrix.
A better understanding of the error term could also help in extending our approach to problems with fine-grained lower bounds based on APSP.

\subparagraph{Approximate learning algorithm.}

Our hardness proof of the learning problem in Section~\ref{sec:hardness} shows that we cannot hope to compute an optimal certificate for a collection of graphs exactly.
Naturally, one could resort to looking for an approximately optimal certificate.
The first question to answer is what should be approximated; we could aim for approximation of the error, or allow to compute a larger certificate than the optimal one that we compare against.
Relaxing both of these conditions is known as bicriteria approximation, which is certainly interesting in our setting, since a certificate that is larger by a factor of $c$ only adds at most the same factor $c$ to the running time. 

\begin{open}
    What is the best bicriteria approximation for \Cref{prob:learning}?
\end{open}

Solving instances like the one in our hardness proof requires computing a set of $k$ vertices in a hypergraph that induces the maximum number of edges, that is, Densest $k$-Subgraph (D$k$S) in hypergraphs.
For D$k$S in standard graphs, only an $\bigo(n^{1/4})$-approximation is known~\cite{bhaskara10}, polylogarithmic approximations are ruled out under the Exponential Time Hypothesis~\cite{manurangsi17}, and there are strong reasons to believe that the upper bound could be optimal~\cite{jones23}.

However, there are several distinguishing details that make our problem different.
First, we are interested in hypergraphs, which are generally harder but also less understood~\cite{applebaum13,chlamtavc17,chlamtavc18}.
Second, in contrast to D$k$S, we want to approximately minimize the number of edges that are \emph{not} in the induced subgraph. This coincides for exact computation but could be very different for multiplicative approximation.
Third, due to the large verification running time, the certificate optimally has size $p = \Theta(n^{0.83})$, which corresponds to the number of vertices in the subgraph; however, the hardness hypotheses used for D$k$S only apply to the regime $k \le \sqrt{n}$~\cite{chlamtavc17,bhaskara10}.
Finally, it is not clear how bicriteria approximation could help and only little is known for standard D$k$S~\cite{khuller09}.
Understanding what can be done and what cannot for our learning problem is hence an interesting challenge.

\subsection{Notation and preliminaries}
We only work with directed weighted graphs $G = (V,E,w_G)$ with $n = |V|$ vertices. For $(u, v) \in E$, we use $w_G(u,v)$ to denote the weight of the edge $(u,v)$. For $(u, v) \in V^2 \setminus E$ with $u \neq v$, we let $w_G(u,v) \coloneqq \infty$, and for $u = v$ we set $w_G(u,v) \coloneqq 0$.
We write $d_G(u,v)$ (or just $d(u,v)$ when $G$ is clear from context) to denote the length of a shortest path from $u$ to $v$.
For $u,v,w \in V$, we say that an algorithm \textbf{\emph{relaxes}} a matrix entry $D[u,v]$ with $D[u,w] + D[w,v]$ to refer to setting $D[u,v]$ to the 
minimum of the two values.

All our algorithms only add, subtract, and compare edge weights, so they work equally well in the word RAM model with polynomially-bounded integer edge weights and in the real RAM model with real edge weights.

\subsection{Organization} 
We discuss Stages 1, 2, and 3 of the algorithm in Sections \ref{sec:distance-estimates}, \ref{sec:verification}, and \ref{sec:fixing}, respectively.
See Section~\ref{sec:negative-weights} for a discussion on how to handle negative edge weights.
In Section~\ref{sec:hardness}, we prove that the learning problem is NP-hard to solve exactly.
Section~\ref{sec:empirical} presents some empirical evidence that one can hope for small prediction error in practice.

\section{Stage 1: Computing distance estimates}\label{sec:distance-estimates}

In the first stage of the algorithm, we compute the initial distance estimates from the certificate.
The distance estimate matrix can also be defined by specific properties, as given in \Cref{lem:distance-estimates}.
Intuitively, we use the certificate to restrict the set of possible relaxations. The distance estimates are then the smallest matrix that can be obtained by using only these relaxations.
To avoid having to do a relaxation multiple times, we compute distance estimates in increasing order.

\begin{lemma}\label{lem:distance-estimates}
    There is an algorithm that, given a graph $G$ and a certificate $C$ with $|C[u,v]| = q$ for all $u,v\in V$, computes the element-wise maximum matrix $D$ that satisfies, for all $u,v \in V$, \begin{enumerate}
        \item $d(u,v) \le D[u,v]$, and
        \item $D[u,v] = \min\{w_G(u,v), \min_{w \in C[u,v]} D[u,w] + D[w,v]\}$.
    \end{enumerate}
    The algorithm runs in time $\bigo(n^2(q + \log n))$.
\end{lemma}
    \begin{algorithm}[t]
    \caption{The algorithm for Stage 1, as in \Cref{lem:distance-estimates}}\label{alg:distance-estimates}
    \KwInput{Graph $G$ and certificate $C$}
    \KwOutput{Distance estimate matrix $D$}
    \Def{\FComputeEstimates{$G$, $C$}}{ 
        $D[i,j] \gets w_G(i,j)$\tcp*{$w_G(i,i) = 0$, $w_G(i,j) = \infty$ for $(i,j) \notin E$}
        $Q \gets$ priority queue with all $(i,j) \in V^2$ with priority $D[i,j]$\;
        \While{$Q$ is not empty}{
            $(i,j) \gets Q.\text{extractMin}()$\;
            \ForEach{$k$ such that $j \in C[i,k]$}{
                $D[i,k] \gets \min \{D[i,k], D[i,j] + D[j,k]\}$\;
                $Q.\text{decreaseKey}((i,k), D[i,k])$\;
            }
            \ForEach{$k$ such that $i \in C[k,j]$}{
                $D[k,j] \gets \min \{D[k,j], D[k,i] + D[i,j]\}$\;
                $Q.\text{decreaseKey}((k,j), D[k,j])$\;
            }
        }
        \Return{$D$}
    }
    \end{algorithm}

\begin{proof} 
    Our algorithm (see \Cref{alg:distance-estimates}) works in a Dijkstra-like manner by computing estimates in increasing order and relaxing other entries with just computed ones.
    The algorithm initializes $D$ with $D[u,v] \coloneqq w_G(u,v)$. It keeps a priority queue $Q$ that is initially filled with every pair $(u,v) \in V^2$ with priority $D[u,v]$. While $Q$ is not empty, it extracts a pair $(u,v)$ with the smallest priority from $Q$. Then, for every $w$ such that $v \in C[u,w]$ it relaxes $D[u,w]$ with $D[u,v] + D[v,w]$. Also, for $w$ such that $u \in C[w,v]$, it relaxes $D[w, v]$ with $D[w, u] + D[u,v]$. If an entry in $D$ changes, it updates its priority in $Q$ accordingly.

    \subparagraph*{Running time.}
    Notice that for every certificate entry $w \in C[u,v]$, we make only two relaxations, one when extracting $(u,w)$ and one when extracting $(w,v)$. Hence, there are at most $n^2q$ relaxations in total. Additionally, there are $n^2$ extractions from the priority queue. With a suitable priority queue, such as Fibonacci heaps~\cite{fredman1987fibonacci}, which decreases keys in time $\bigo(1)$ and extracts the minimum in time $\bigo(\log n)$,
    this adds up to the claimed running time.

    \subparagraph*{Correctness.}
    The first property follows from the fact that we initialize with edge weights and from then on only relax current distance values.
    Therefore, every element is also at most its edge weight from the initialization and relaxation steps.
    So the first part $D[u,v] \le w_G(u,v)$ of the second property holds too.
    
    Remember that edge weights are non-negative. Therefore, a relaxation can never decrease a distance entry below already extracted entries. This ensures that the pairs are extracted non-decreasingly.
    Thus, if we show that whenever a pair $(u,v)$ is extracted from $Q$ with $D[u,v] < w_G(u,v)$ we have $D[u,v] = \min_{w \in C[u,v]} D[u,w] + D[w,v]$, then it also still has to hold after the algorithm is completed.
    Remember that we can assume that all weights are non-negative, as discussed in Section~\ref{sec:negative-weights}.
    Therefore, if $D[u,v] > D[u,w] + D[w,v]$ for some $w \in C[u,v]$ when extracted, both $(u,w)$ and $(w,v)$ must have been extracted before. Then $D[u,v]$ should have been relaxed with this quantity, a contradiction. 

    Furthermore, notice that whenever an entry of $D$ is relaxed it is decreased to the maximum possible value it could have. Thus, the matrix is the unique element-wise maximum that fulfills the requirements.
\end{proof}

\section{Stage 2: Verification of the distance estimates}\label{sec:verification}

In the next step, we give a verification algorithm that verifies for each $(u,v) \in V^2 \setminus \jump{G}{C}{p}$ that it cannot be further relaxed in the matrix $D$, obtained from \Cref{lem:distance-estimates}. The algorithm might coincidentally verify more pairs, so formally, it outputs a set $U \subseteq \jump{G}{C}{p}$ and guarantees $D[u,v] = \min_{w \in V} D[u,w] + D[w,v]$ for all $(u,v) \in V^2 \setminus U$.
To achieve this, we follow the co-nondeterministic algorithm for Exact Triangle from~\cite[Theorem~9.7]{ChanWX23}.
They introduce a technique called ``Fredman's trick meets dominance product''. In this case, it allows us to speed up counting the number of smaller detour vertices with fast matrix multiplication.
We adapt the exact distances in their algorithm to minimum distances for shortest paths. Furthermore, we replace their non-deterministic hitting set and non-deterministic dominance product computation with efficient deterministic counterparts.

\begin{lemma}\label{lem:verify-cert} 
    There is an algorithm that is given a graph $G$, a certificate $C$ of size $|C[u,v]| = q$ for all $u,v \in V$, and a parameter $p$. It computes $U \subseteq \jump{G}{C}{p}$ such that for all $(u,v) \in V^2 \setminus U$ we have $D[u,v] = \min_{w \in V} D[u,w] + D[w,v]$.
    
    The algorithm runs in time $\otilde\left(n^2q + n^{3.66} / p \right)$, or $\bigo\left(n^{2.83}\right)$ for $p^* = n^{0.83}$ and $q = \bigo(p)$.
\end{lemma}

\begin{algorithm}[t]
    \caption{The algorithm for Stage 2, as in \Cref{lem:verify-cert}, adapted from~\cite{ChanWX23}}\label{alg:verify-cert}
    
    \KwInput{Distance estimate matrix $D$, certificate $C$, and a parameter $p$}
    \KwOutput{Set of unverified pairs $U$}

    \Def{\FVerifyEstimates{$D$, $C$, $p$}}{
        \ForEach{$(i,j) \in V^2$}{
            \lForEach{$k \in C[i,j]$}{
                $\detour{D}{i}{j}{k} \gets D[i,k] + D[k,j] - D[i,j]$
            }
            $C'[i,j] \gets$ the $p$ vertices $k \in C[i,j]$ with the smallest $\detour{D}{i}{j}{k}$\;
        }
        $R \gets$ greedy hitting set for all $C'[i,j]$\;
        \ForEach{$(i,j) \in V^2$}{
            $r_{i,j} \gets \arg \min_{r \in R \cap C'[i,j]} \detour{D}{i}{j}{r}$\;
            $\text{count}_1[i,j] \gets \big|\big\{k \in C[i,j] \colon D[i,k] + D[k,j] < D[i,r_{i,j}] + D[r_{i,j}, j]\big\}\big|$\;
        }
        \ForEach(\tcp*[f]{Compute \# of better witnesses than $r$ for all $i,j$}){$r \in R$}{
            $A[i,j] \gets D[i,j] - D[i,r]$\;
            $B[i,j] \gets D[r,j] - D[i,j]$\;
            $\text{Dom} \gets A \domprod B$ \tcp*{Yuster's dominance product algorithm}
            \lForEach{$i,j \in V$ such that $r = r_{i,j}$}{$\text{count}_2[i,j] \gets \text{Dom}[i,j]$}
        }
        \Return{$\{(i,j) \colon \mathrm{count}_1[i,j] \ne \mathrm{count}_2[i,j]\}$}
    }
\end{algorithm}

\begin{proof}
    The algorithm is adapted from~\cite{ChanWX23} and works as follows (see \Cref{alg:verify-cert}). 
    \begin{itemize}
        \item For every $u,v \in V, w \in C[u,v]$, compute $\detour{D}{u}{v}{w}$. Let $C'[u,v]$ be the $p$ vertices in $C[u,v]$ with the smallest detour.
        
        \item Compute a set $R \subseteq V$ such that for all $u,v \in V$ the intersection $R \cap C'[u,v]$ is non-empty, that is, a hitting set for all $C'[u,v]$. To do so, run the greedy algorithm from Chapter 2 in~\cite{vazirani2001approximation}. We will show that $R$ has size $\bigo(n \log^2 n / p)$.
        
        \item For all $u,v$, fix an arbitrary representative $r_{u,v} \in R$ that appears in $C'[u,v]$. Any such choice works but choosing an $r$ with minimum $\detour{D}{u}{v}{r}$ minimizes the size of $U$ and hence verifies more distances.
        
        \item For all $u,v\in V$, count the number of $w \in C[u,v]$ such that $D[u,w] + D[w,v] < D[u,r_{u,v}] + D[r_{u,v},v]$ using brute force. 
        
        \item For all $u,v \in V$, count the number of $w \in V$ such that $D[u,w] + D[w,v] < D[u,r_{u,v}] + D[r_{u,v},v]$ using ``Fredman's trick meets dominance product''.
        Specifically, iterate over all $r \in R$ and count the number of $w \in V$ with $D[u,w] + D[w,v] < D[u,r] + D[r,v]$.
        To do so, construct $A$ and $B$ with $A[i,j] \coloneqq D[i,j] - D[i,r]$ and $B[i,j] \coloneqq D[r,j] - D[i,j]$.\footnote{If the term $\infty - \infty$ arises in this computation, we replace it with $+\infty$ if in $A$ and with $-\infty$ if in $B$.} Compute the dominance product $\text{Dom}$ with $\text{Dom}[i,j] = |\{k \colon A[i,k] < B[k,j]\}|$. By construction, we have \[A[i,k] < B[k,j] \Leftrightarrow D[i,k] + D[k,j] < D[i,r] + D[r,j].\]
        
        \item Output the set $U$ of pairs $(u,v)$ for which the brute force and dominance product counts are different.
    \end{itemize}

    \subparagraph*{Running time.} Since we are given $n^2$ subsets of size $q$ from a universe of size $n$, a uniformly random set of size $\bigo(n\log n/p)$ hits all of them with non-zero probability by a standard argument (see, e.g.,~\cite[Chapter~1]{alon2016probabilistic}). Hence, the optimal hitting set also has size at most $\bigo(n\log n/p)$. The greedy hitting set algorithm from Chapter 2 in~\cite{vazirani2001approximation} is a $\log n$-approximation and runs in near-linear time. We get $|R| = \bigo(n\log^2n/p)$.

    Computing the detours, the greedy hitting set, and the brute force counts takes time $\bigo(n^2q)$.
    Yuster's algorithm~\cite{Yuster09} computes dominance product in time $\bigo(n^{4-r} + n^{\omega(1,r,1) + o(1)})$ for any $r$. Here $\omega(a,b,c)$ is the rectangular matrix multiplication exponent such that it takes time $n^{\omega(a,b,c) + o(1)}$ to multiply a $n^a \times n^b$ with a $n^b \times n^c$ matrix.
    With the optimal choice of $r$ and the best known algorithm for rectangular matrix multiplication~\cite{WilliamsXXZ24, Brand_complexity_tool}, computing one dominance product takes time $\bigo(n^{2.658})$.
    In total, this gives the claimed running time.

    \subparagraph*{Correctness.}
    First, consider $(u,v) \in V^2 \setminus \jump{G}{C}{p}$, so for all $w \in V \setminus C[u,v]$, we have \[|\{c \in C[u,v]\colon \detour{D}{u}{v}{c} \le \detour{D}{u}{v}{w}\}| \ge p.\]
    Since $r_{u,v}$ is within the $p$ elements of $C[u,v]$ with minimum $\detourOp_{D,u,v}$, for all $w \in V \setminus C[u,v]$, we have $\detour{D}{u}{v}{r_{u,v}} \le \detour{D}{u}{v}{w}$.
    So, both counts in the algorithm must be the same and $(u,v) \in V^2 \setminus U$.
    
    Furthermore, if both counts are the same notice that \[\min_{w \in C[u,v]} D[u,w]+D[w,v] = \min_{w \in V} D[u,w]+D[w,v].\] So, by \Cref{lem:distance-estimates}, the condition holds for all $(u,v) \in V^2 \setminus U$.
\end{proof}

We note that it would be possible to further speed up the verification part by using a co-nondeterministic algorithm for dominance product from~\cite{ChanWX23}. However, that algorithm would need additional information about the dominance product included in the certificate. This would require more complex predictions and introduce additional sources of error, while not significantly improving the result. The best running time one could reach this way is $\otilde(n^{(6 + \omega) / 3})$, or $\bigo(n^{2.80})$ with the current bound on the matrix multiplication exponent~\cite{alman2025more}.

Intuitively, for every distance entry $D[u,v]$ that is wrong but successfully verified, there must be another wrong entry $D[u,w]$ or $D[w,v]$ that the computation relied on. This allows us later in the final algorithm to trust correctly verified entries until we discover such a witness with incorrect distance. This is formalized in the next lemma.

\begin{lemma}\label{lem:subpath_in_jumpset}
    For all $(u,v)$ with $D[u,v] \ne d(u,v)$, there is $(a,b) \in U$, potentially equal to $u$ or $v$, with $d(u,v) = d(u,a) + d(a,b) + d(b,v)$.
\end{lemma}
\begin{proof}
    Consider a pair $(u,v)$ with $D[u,v] = \min_{w \in V} D[u,w] + D[w,v]$ and $D[u,v] \ne d(u,v)$. Let $P$ be a shortest ($u$,$v$)-path with the fewest number of edges. If $P = (u,v)$, this is a contradiction to the assumption that $D[u,v]$ is not the correct distance by \Cref{lem:distance-estimates}.

    Otherwise, let $w$ be an intermediate vertex on $P$. Since $D[u,v] \le D[u,w] + D[w,v]$, at least one of $D[u,w]$ and $D[w,v]$ is also not the correct distance, say $D[u,w]$.
    If $(u,w) \in U$, we are done. Otherwise, we have $D[u,w] = \min_{x \in V} D[u,x] + D[x,w]$ and we can repeat the same argument with $(u,w)$ instead of $(u,v)$. This way, the number of edges on the fewest-hop shortest path decreases by at least one. Hence at some point we arrive at a pair in $U$.
\end{proof}

\section{Stage 3: Fixing the distance estimates}\label{sec:fixing}

In this section, we give the final stage of our algorithm and prove \Cref{thm:full-alg}. It works similar in spirit to Stage 1 as we fix entries in increasing order of distance. We rely on the following observation when computing the length of a shortest path from $u$ to $v$. Since all subpaths have smaller or equal distance, their values are correct. Furthermore, this means that if $D[u,v]$ is unchanged and successfully verified, we can skip all relaxations of $D[u,v]$.

\mainThm*

\begin{algorithm}[t]
    \caption{The algorithm for Stage 3, as in \Cref{thm:full-alg}}\label{alg:fixing}
    \KwInput{Distance estimate matrix $D$, unverified set $U$}
    \KwOutput{Correct distance matrix $D$}

    \Def{\FFixMistakes{$D$, $U$}}{
        $M \gets U$ \tcp*{Marked pairs}
        $Q \gets$ priority queue with all $(i,j) \in V^2$ with priority $D[i,j]$\;
        \While{$Q$ is not empty}{
            $(i,j) \gets Q.\text{extractMin}()$\;
            \uIf(\tcp*[f]{If marked, relax everything and mark changed}){$(i,j) \in M$}{
                \ForEach{$k \in V$}{
                    \If{$D[i,j] + D[j,k] < D[i,k]$}{
                        $D[i,k] \gets D[i,j] + D[j,k]$\;
                        $Q.\text{decreaseKey}((i,k), D[i,k])$\;
                        $M \gets M \cup \{(i,k)\}$\;
                    }
                    \If{$D[k,i] + D[i,j] < D[k,j]$}{
                        $D[k,j] \gets D[k,i] + D[i,j]$\;
                        $Q.\text{decreaseKey}((k,j), D[k,j])$\;
                        $M \gets M \cup \{(k,j)\}$\;
                    }
                }
            }
            \Else(\tcp*[f]{If not marked, only relax marked}){
                \ForEach{$k$ such that $(i,k) \in M$}{
                    $D[i,k] \gets \min \{D[i,k], D[i,j] + D[j,k]\}$\;
                    $Q.\text{decreaseKey}((i,k), D[i,k])$\;
                }
                \ForEach{$k$ such that $(k,j) \in M$}{
                    $D[k,j] \gets \min \{D[k,j], D[k,i] + D[i,j]\}$\;
                    $Q.\text{decreaseKey}((k,j), D[k,j])$\;
                }
            }
        }
        \Return{$D$}
    }
\end{algorithm}

\begin{proof}
    Remember that we can assume that all edge weights are non-negative by a standard preprocessing step, as discussed in Section~\ref{sec:negative-weights}.
    First, we run the algorithm from \Cref{lem:distance-estimates} to compute the matrix $D$ from $C$.
    Then, we use \Cref{lem:verify-cert} to compute the set $U$. We call all pairs $(u,v) \in U$ marked.
    During the algorithm we will relax $D$ to turn it into the correct distance matrix for $G$. If at
    any point a distance changes, we also mark the corresponding pair. The remaining algorithm (see \Cref{alg:fixing})
    is similar to the one in \Cref{lem:distance-estimates} but distributes the total amount of work only on marked pairs.
    \begin{itemize}
        \item Initialize a priority queue $Q$ with all pairs $(u,v) \in V^2$ and priorities $D[u,v]$.
        \item While $Q$ is nonempty, extract the pair $(u,v)$ with minimum $D[u,v]$, and do the following: \begin{itemize}
            \item If $(u,v)$ is marked, iterate over all $w \in V$. Relax $D[u,w]$ with $D[u,v] + D[v,w]$. Relax $D[w,v]$ with $D[w,u] + D[u,v]$. If anything changed, mark the corresponding entry and update the key in $Q$.
            \item If $(u,v)$ is not marked, iterate only over $w \in V$ where $(u,w)$ or $(w,v)$ is marked. Relax $D[u,w]$ with $D[u,v] + D[v,w]$. Relax $D[w,v]$ with $D[w,u] + D[u,v]$. Update $Q$.
        \end{itemize}
        \item When $Q$ is empty, return $D$ as the distance matrix of $G$.
    \end{itemize}

    \subparagraph*{Running time.}
    Let $(u,v)$ be a pair of vertices that does not contribute to the error $\eta$. Therefore, we have $D[u,v] = d(u,v)$ and $(u,v) \in V^2 \setminus \jump{G}{C}{p}$.
    By \Cref{lem:verify-cert}, every pair $(u,v) \in V^2 \setminus \jump{G}{C}{p}$ is also not in $U$. It is hence not marked initially. Since $D$ always contains an upper bound to the true distances in $G$, the entry $D[u,v]$ will never be changed. So, $(u,v)$ will not be marked during the algorithm. This shows that at most $\eta$ pairs will become marked during the whole algorithm. 

    Every marked pair causes the relaxation of $2n$ pairs when extracted from $Q$. Additionally, it can be relaxed by at most $2n$ unmarked pairs when they are extracted from $Q$. A marked pair is hence responsible for $\bigo(n)$ cost during the whole algorithm. All extractions from $Q$ cost $\bigo(n^2 \log n)$ using Fibonacci heaps~\cite{fredman1987fibonacci}. Therefore, the main part of the algorithm takes time $\bigo(\err{G}{C}{p} \cdot n + n^2 \log n)$.
    The distance estimate and verification parts take time $\bigo(n^2(q + \log n) + n^{3.66} / p)$ by \Cref{lem:verify-cert,lem:distance-estimates}. In total, this gives the final running time.

    \subparagraph*{Correctness.}
        We show correctness of the remaining algorithm by the following invariants.
        \begin{romanenumerate}
            \item \label{prop-corr}Every pair $(u,v)$ that is already extracted from $Q$ has $D[u,v] = d(u,v)$.
            
            \item \label{prop-subpath}For every pair $(u,v)$ with $D[u,v] \ne d(u,v)$, there is a shortest path from $u$ to $v$ that contains a (potentially equal) subpath from $x$ to $y$ such that $(x,y)$ is marked and in $Q$.
            
            \item \label{prop-marked} For every marked pair $(u,v) \in Q$, the current distance is at most the shortest path constructed from two extracted paths, i.e., \[D[u,v] \le \min_{(u,w), (w,v) \in V^2 \setminus Q} d(u,w) + d(w,v).\]
        \end{romanenumerate}

        In the beginning, (\Cref{prop-subpath}) is satisfied by \Cref{lem:subpath_in_jumpset}, the other two are trivially satisfied.
         In each step, we extract one pair $(u,v)$ with current minimum $D[u,v]$. 
        Suppose $(u,v)$ is currently extracted. First, we prove that every pair $(u', v')$ with $d(u',v') < d(u,v)$ is already extracted.
        Suppose there is another pair $(u',v')$ still in $Q$ with $d(u',v') < d(u,v) \le D[u,v] \le D[u',v']$. By (\Cref{prop-subpath}), there is a marked subpath of a shortest $(u',v')$-path still in $Q$.
        Let $(a,b)$ be a minimal such subpath.
        Since all subpaths of $(a,b)$ are either extracted or unmarked but with no marked subpath in $Q$, all subpaths of $(a,b)$ must be correct.
        If there is one such path still in $Q$, it must be shorter and hence extracted before $(u',v')$ and $(u,v)$.
        If none are in $Q$, by (\Cref{prop-marked}), $D[a,b]$ is correct.
        But then \[D[a,b] = d(a,b) < d(u',v') < d(u,v) \le D[u,v],\] and $(a,b)$ should have been extracted before $(u,v)$, a contradiction.

        For (\Cref{prop-corr}), suppose we extract $(u,v)$ with $D[u,v] > d(u,v)$ for the sake of contradiction. Notice that every subpath of a shortest ($u$,$v$)-path has weight at most $d(u,v)$ and must hence be already extracted and, by (\Cref{prop-corr}), also correctly computed.
        By (\Cref{prop-subpath}), there must be a marked subpath $(x,y)$ of $(u,v)$ in $Q$.
        But then again, by (\Cref{prop-corr}) and (\Cref{prop-marked}), $D[x,y] = d(x,y)$ and it should be extracted before $(u,v)$.

        Next, we show that (\Cref{prop-subpath}) is still satisfied. The only way it can no longer be satisfied is if the extracted $(u,v)$ was marked and now there is some unmarked $(x,y)$ that does not have the correct distance and no marked subpath. Then, there is a shortest path via $x \to u \to v \to y$.
        Notice that $D[x,u]$ must also be correct, because otherwise $(x,u)$ is marked or there is another marked subpath of $(x,u)$ and hence $(x,y)$. Similarly, $D[v,y]$ must be correct.
        Since $D[x,v]$ is relaxed with $D[x,u] + D[u,v]$ now, it must also be correct, as well as $D[u,y]$.
        This means that $D[x,u]$ and $D[u,y]$ are both correct and unmarked, so they were always correct and unmarked. But then, by assumption that $D[x,y] > d(x,y)$, we know that $(x,y)$ must be in $U$ by \Cref{lem:verify-cert} and hence be marked, a contradiction.

        Furthermore, (\Cref{prop-marked}) clearly still holds by definition of the algorithm.
        After the execution is complete, $D$ holds all correct distances by (\Cref{prop-corr}) and the algorithm is correct.
\end{proof}

\section{Handling negative edge weights}\label{sec:negative-weights}

Our main algorithm crucially relies on nonnegativity of edge weights. In this section, we describe how to preprocess the graph to ensure that this condition holds.
Assume there is no cycle of negative total weight, since shortest paths are not well-defined otherwise.
We employ a standard trick in the area of shortest paths~(see, e.g.,~\cite{Fischer25}), originally due to Johnson~\cite{Johnson77}.
Let graph $G = (V,E,w)$ be an edge-weighted graph and let $\phi: V \to \mathbb{R}$ be a \emph{potential function}. For any edge $(u,v) \in E$, define the \emph{$\phi$-adjusted} edge weights $w_\phi(u,v) \coloneqq w(u,v) + \phi(u) - \phi(v)$.

We state a few observations about this definition. First, in the $\phi$-adjusted graph $G_\phi = (V, E, w_\phi)$, the shortest path lengths stay the same up to the potential of the endpoints, i.e., $d_G(u,v) = d_{G_\phi}(u,v) - \phi(u) + \phi(v)$.
In fact, for fixed $u, v \in V$, the length of any ($u$,$v$)-path gets preserved (when switching between $G$ and $G_\phi$) up to the same offset $\phi(u)-\phi(v)$.
Johnson's trick states that we can always compute a potential $\phi$ such that $w_\phi$ is non-negative.

\begin{observation}[Johnson's trick~\cite{Johnson77}]\label{obs:johnson}
    By choosing $\phi(v) = \min_{u \in V} d(u,v)$, all $\phi$-adjusted weights $w_\phi$ are non-negative.
\end{observation}

Note that to compute a potential $\phi$ as in \Cref{obs:johnson}, we can simply add one vertex $s$ to $V$ with weights $w(s,v) = 0$ for all $v \in V$ and run any single-source shortest path algorithm that works for negative weights. 

By computing such a potential $\phi$ in the beginning, running our algorithms in \Cref{sec:distance-estimates,sec:verification,sec:fixing} on the $\phi$-adjusted graph, and removing the offset in the end, we can hence guarantee non-negative edge weights. Note that we do not change the certificate. It remains to argue that the prediction error and thus the running time is unchanged.
For Stage 1, consider the matrices $D$ and $D_\phi$ that satisfy the guarantees of \Cref{lem:distance-estimates} for $G$ and $G_\phi$ respectively.
We can verify inductively that $D_\phi[u,v] = D[u,v] - \phi(u) + \phi(v)$ by noticing that for any $w \in V$ we have
\[D[u,w] - \phi(u) + \phi(w) + D[w,v] - \phi(w) + \phi(v) = D[u,w] + D[w,v] - \phi(u) + \phi(v).\]
By the same observation, the detour function in $G_\phi$ remains completely unchanged, i.e. $\detour{D_\phi}{u}{v}{w} = \detour{D}{u}{v}{w}$ since all offsets cancel. 
This already guarantees that $\jump{G_\phi}{C}{p} = \jump{G}{C}{p}$ and $\err{G_\phi}{C}{p} = \err{G}{C}{p}$ and the error is unaffected.

The overhead for running a negative-weight single-source shortest path algorithm depends on the algorithm we choose.
One option is to use one of the near-linear-time scaling-based algorithms~\cite{bernstein2025negative,Fischer25}. These algorithms only work for integer weights, but lead to no overhead in the running time of \Cref{thm:full-alg} for any parameters $q$ and $p$, unless the edge weights are at least exponential.
Another option is to use one of the recent algorithms for negative real weights~\cite{Fineman24,Huang25}, with the fastest such algorithm running in time $\otilde(mn^{4/5}) \le \bigo(n^{2.83})$. For optimal parameter choices $p = \Theta(n^{0.83})$ and $q = \Theta(p)$, our algorithm hence still runs in time $\bigo(n^{2.83} + \eta n)$, even for real and possibly negative edge weights.

\section{The learning problem is NP-hard}\label{sec:hardness}

This section is dedicated to the NP-hardness proof of the learning problem that corresponds to our error measure $\eta$.
In terms of the certificate size, we prove hardness for almost the full range of possible sizes for the certificate sets.
However, the proof holds only for the regime with many instances for which we learn a certificate at the same time.

\begin{theorem}\label{thm:learning-hard}
    For every $\gamma, \delta \in (0,1)$ with $\gamma \le \delta$, \Cref{prob:learning} is NP-hard for $p = n^\gamma$ and $q = n^\delta$ with $k = \Omega(q^2)$. The same holds for $p = n- n^\delta$ and $q = n - n^\gamma$.
\end{theorem}
\begin{proof}
    \begin{figure}
        \centering
        \includegraphics[width=0.5\linewidth,page=1]{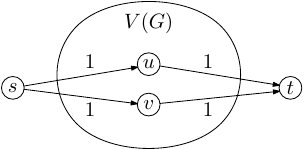}
        \caption{A sketch of the reduction for \Cref{thm:learning-hard}. For every edge $e_i = \{u,v\} \in E(G)$, we create this graph $G_i$. With an optimal choice of the certificate $C$, this graph induces error 0 if and only if $\{u,v\} \subseteq C[s,t]$.}
        \label{fig:learning-hard}
    \end{figure}

    We reduce from $(q - 2)$-Clique, where we are given a graph $G$ with $|V(G)| = n$ and are supposed to decide if it contains a clique on $q - 2$ vertices. 
    Notice that since $q = n^\delta$, or $q = n - n^\gamma$, that problem is NP-hard by a standard reduction from the normal Clique problem (see, e.g., Chapter 7 in~\cite{sipser1996introduction}).

    For every edge $e_i = \{u,v\} \in E(G)$, we create a graph $G_i$ with the following vertices, edges, and weights (see \Cref{fig:learning-hard}). Set 
    \begin{itemize}
        \item $V(G_i) \coloneqq V(G) \cup \{s,t\}$,
        \item $E(G_i) \coloneqq \{(s,u), (u,t), (s,v), (v,t)\}$, and
        \item for all $e \in E(G_i)$, set $w(e) \coloneqq 1$.
    \end{itemize}
    
    We claim that $G$ contains a clique of size $q-2$ if and only if there is a certificate $C$ with $|C[u,v]| = q$ for the $|E(G)|$ graphs with error at most $|E(G)| - \binom{q-2}{2}$.

    Suppose $G$ contains a clique $S = \{s_1, \dots, s_{q-2}\}$. Consider the certificate that is defined as follows. 
    \begin{itemize}
        \item We set $C[s,t] \coloneqq (s, t, s_1, \dots, s_{q-2})$.
        \item All other lists $C[u,v]$ are filled with $u$, $v$ and $q-2$ other vertices with minimum index.
    \end{itemize}

    Let $e_i = \{u,v\} \in E(G)$. We are left with the following situation. The only finite distances in $G_i$ are $d_{G_i}(s,u)$, $d_{G_i}(s,v)$, $d_{G_i}(u,t)$, $d_{G_i}(v,t)$, and $d_{G_i}(s,t)$. All entries of the distance estimates $D_i$ for $G_i$ are automatically correct from the initialization except for $D_i[s,t]$. The entry $D_i[s,t]$ is correctly filled if and only if at least one of $u$ and $v$ is in $C[s,t]$. Finally, $(s,t)$ does not count into the error term $\err{G_i}{C}{p}$ if and only if \emph{both of $u$ and $v$} and $s$ and $t$ are in $C[s,t]$, since all other intermediate vertices yield an infinite detour.

    This allows us to conclude that the error as stated in \Cref{prob:learning} is exactly \[ |\{\{u,v\} \in E(G) \colon \{u,v\} \not\subseteq C[s,t]\}| = |E(G)| - \binom{q-2}{2}. \]

    Consider now the other direction and suppose we have a certificate $C$ that results in $D_i$ with a total error of at most $|E(G)| - \binom{q-2}{2}$. Without loss of generality, we can assume that for every $a,b \in V(G_i)$ except for $a = s$ and $b = t$ the list $C[a,b]$ are chosen optimally. That is because $D_i[a,b]$ is always correct and the lists have no influence on other values of $D_i$. Then, $(a,b)$ cannot contribute to the $i$-th error term. Hence, it is enough to focus on $s$ and $t$. As noticed before, $(s,t)$ for $e_i = \{u,v\}$ will contribute to the error term $\err{G_i}{C}{p}$ if and only if $\{s,t,u,v\} \not\subseteq C[s,t]$. We can assume that $\{s,t\} \subseteq C[s,t]$, otherwise the error is at least $|E(G)|$. Furthermore, by the supposed error term, for at least $\binom{q-2}{2}$ edges $\{u,v\}$, both $u$ and $v$ appear in $C[s,t]$. Thus, $C[s,t] \setminus \{s,t\}$ must be a clique of size $q-2$.

    Since the reduction function is clearly polynomial, this proves the hardness. For the parameters, notice that clique becomes trivial if there are less than $q-2$ edges in $G$.
\end{proof}

\section{Empirical estimation of prediction error}
\label{sec:empirical}

In order to understand whether in practice one can expect our prediction error to be substantially smaller than $n^2$, we ran a simulation on some real-world data. Our code is available at \url{https://github.com/adampolak/warm-starting-apsp}.

We use two datasets consisting of network latencies between fixed sets of nodes over multiple time slices~\cite{ZhuLNLZ17}, available at \url{https://github.com/uofa-rzhu3/NetLatency-Data/}. We remark that, even though road networks are a textbook example application of APSP, in reality they are extremely sparse, and hence simply running Dijkstra from each node already gets one very close to a quadratic running time, which is a hard theoretical lower bound because of the output size. Network latency graphs constitute an example of real-world dense graphs on which people compute APSP~\cite{ZhuLNLZ17}.

For each of the two datasets, we independently sampled $10$ graph snapshots. For each sampled graph, we computed a ground-truth certificate by solving APSP exactly. We then used this certificate as a prediction for the next graph snapshot in the time sequence. We calculated the prediction error for our algorithm with parameters set to $p=\frac{1}{2}\sqrt{n}$ and $q=2\sqrt{n}$.

The first dataset, called Seattle, consists of $t=688$ snapshots of a graph on $n=99$ nodes. The average prediction error that we calculated equals $\eta_{\text{mean}} = 121.2$ (with the standard deviation $\sigma=67.8$), which is only slightly more than $n$ and significantly less than the worst-case upper bound of $n^2$.

The second dataset, called PlanetLab, consists of $t=18$ snapshots of a graph on $n=490$ nodes. The average prediction error that we calculated equals $\eta_{\text{mean}} = 12\,810.8$ (with the standard deviation $\sigma=4887.4$), which is close to $n^{1.5}\approx10\,846.6$.

We conclude that it is not unreasonable to hope for $\eta \ll n^2$ in practice. Of course, it would be more desirable to not only calculate the prediction error but also to implement our algorithm and directly measure its running time. An obstacle to that is that our algorithm requires a fast matrix multiplication algorithm, and those theoretically fast algebraic algorithms are famously known for being impractical. This does not necessarily mean that our algorithm is hopelessly impractical as well. Indeed, given that matrix multiplication lies at the core of deep neural networks, people put a lot of engineering effort, both at software and hardware levels, to make it fast in practice. Such optimizations could also make an implementation of our algorithm competitive, but this is beyond the scope of this theory paper.

\bibliography{main}

\end{document}